\documentclass[10pt,conference]{IEEEtran}
\usepackage{graphicx}
\usepackage{cite}
\usepackage{url}
\usepackage{floatrow}
\usepackage{stfloats}
\usepackage{amsmath,amsthm,amsfonts,amssymb}
\usepackage{array}
\usepackage{url}
\usepackage{graphicx}
\usepackage{cite}
\usepackage{color}
\usepackage{tikz}
\usepackage{circuitikz}
\usepackage{mathtools}

\newtheorem{lem}{Lemma}
\newtheorem{res}{Result}

\floatstyle{plaintop}
\restylefloat{table}
\IEEEoverridecommandlockouts
\belowdisplayskip=\abovedisplayskip
\begin{document}

\title{High SNR Probabilities of Continuous Fluid Antenna Systems in Ricean Environments

\author{\IEEEauthorblockN{Amy S. Inwood\IEEEauthorrefmark{1}, Peter J. Smith\IEEEauthorrefmark{2}, Rajitha Senanayake\IEEEauthorrefmark{3}, and Michail Matthaiou\IEEEauthorrefmark{1}}

    \IEEEauthorblockA{\IEEEauthorrefmark{1}Centre for Wireless Innovation (CWI), Queen’s University Belfast, Belfast BT3 9DT, U.K}
    \IEEEauthorblockA{\IEEEauthorrefmark{2}School of Mathematics and Statistics, Victoria University of Wellington, Wellington, New Zealand}

    \IEEEauthorblockA{\IEEEauthorrefmark{3}Department of Electrical and Electronic Engineering, University of Melbourne, Melbourne, Australia}
    
    \IEEEauthorblockA{Email: (a.inwood, m.matthaiou)@qub.ac.uk, peter.smith@vuw.ac.nz, rajitha.senanayake@unimelb.edu.au}
    
    }

\thanks{This work was supported in part by the U.K. Engineering and Physical Sciences Research Council (EPSRC) (grant No. EP/X04047X/1). The work of M. Matthaiou has received funding from the European Research Council (ERC) under the European Union’s Horizon 2020 research and innovation programme (grant agreement No. 101001331).}}

\maketitle
\begin{abstract}
We consider a single-user (SU) continuous fluid antenna system (CFAS) employing matched filtering (MF) operating over a Ricean fading channel. Focusing on the upper tail of the received signal-to-noise ratio (SNR) distribution (the high SNR probability (HSP)), we derive accurate approximations for the HSP in 1, 2, and 3 dimensions using the expected Euler characteristic (EEC), presenting the first analytical results for a CFAS in a Ricean environment. In the process, we provide the first closed-form expression for the Euler characteristic density of a non-central \(\chi_2^2\) random field. We then examine the impact of the Ricean K-factor on the CFAS performance, emphasizing the critical role of channel variations in achieving a strong HSP.
\end{abstract}
\vspace{-0.5em}
\vspace{-0.5em}
\section{Introduction} \label{sec:intro}

Fluid antenna systems (FASs) \cite{kkwong1,kkwong2,FAS_6G} and movable antennas \cite{zhu_moveable_2024,zhu_modeling_2024} are well-established techniques for leveraging the spatial domain to enhance the signal-to-noise ratio (SNR) and mitigate interference. Through flexible antenna positioning, FASs improve the SNR, enhance diversity, and increase reliability with limited hardware requirements. Continuous FASs (CFASs) \cite{psomas,smith_dimensional_2025} maximize these improvements by considering antennas capable of being positioned anywhere in a continuous space. The performance of a FAS depends on channel variations across possible antenna positions. Thus, rich-scattering Rayleigh channels are expected to provide significant benefits, whereas line-of-sight (LoS) environments may yield more limited improvements. In reality, many environments contain both LoS and scattered components, making Ricean fading a more realistic model for such scenarios. Therefore, in this paper, we analyze the performance of a single-user (SU) CFAS employing matched filtering (MF) under Ricean fading.

There is limited work on FASs operating over a Ricean fading channel, with \cite{RISFAS} developing algorithms for a reconfigurable intelligent surface-assisted FAS and \cite{ULFAS} minimizing the transmit power for a multi-user uplink FAS. To date, there has been minimal analytical research on FAS under Ricean fading, so this is our focus. As a performance metric, we analyze the upper tail of the cumulative distribution function (CDF) of the SNR—specifically, the probability that the SNR exceeds a high threshold, which we define as the high SNR probability (HSP). This choice is motivated by prior work on Rayleigh fading \cite{psomas,smith_dimensional_2025}, where the HSP is the only known metric that allows closed-form performance analysis for CFASs. Building on these results, we extend our previous analysis in \cite{smith_dimensional_2025} to Ricean fading, providing the first analytical results for CFAS performance under Ricean fading.

Significant analytical progress has been made for FASs in Rayleigh environments. For a finite number of discrete antenna positions, full SNR distributions have been derived using approximate correlation models \cite{alouini}, block-correlation models \cite{block}, and copulas \cite{cop1,cop2}. As directly extending these methods to CFASs is challenging, the work in \cite{smith_dimensional_2025} takes a different approach by leveraging random field theory \cite{adler} to evaluate the HSP for CFASs in 1D, 2D, and 3D under Rayleigh fading. A key tool in this analysis is the expected Euler characteristic (EEC) \cite{adler}, which provides an asymptotically exact approximation of the HSP. In this paper, we extend this approach to the more complex case of Ricean fading, where the presence of a deterministic LoS component, and thus a non-zero mean, significantly complicates the analysis. More specifically, we make the following contributions:
\begin{itemize}
    \item We present the first closed-form expression for the Euler characteristic density of a noncentral $\chi_2^2$ random field.
    \item We use this expression to derive an accurate approximation to the HSP, based on the received SNR for 1D, 2D and 3D.
    \item We verify the derived approximations with simulations, and investigate the impact of Ricean K-factor on the HSP.
\end{itemize}

\textit{Notation}: Lower boldface letters represent vectors; $\mathbb{E}[\cdot]$ is the statistical expectation; $\mathrm{Var}[\cdot]$ is the variance; $P(A)$ is the probability of event $A$; $\mathcal{CN}(\mu,\sigma^2)$ is a complex Gaussian distribution with mean $\mu$ and variance $\sigma^2$; $\chi_k^2$ is a central chi-squared distribution with $k$ degrees of freedom; $\chi_k^2(\lambda)$ is a noncentral chi-squared distribution with $k$ degrees of freedom and noncentrality parameter, $\lambda$; $Q_\nu(\cdot,\cdot)$ is the Marcum Q-function of order $\nu$, $J_0(\cdot)$ is the zeroth-order Bessel function of the first kind; $I_0(\cdot)$ is the zeroth-order modified Bessel function of the first kind; $\Gamma(\cdot)$ is the gamma function; $(\cdot)^*$ is the complex conjugate; $(\cdot)^T$ represents transpose; $\lfloor \cdot\rfloor$ is the floor operator; $||\cdot||$ is the Euclidean norm and $\mathrm{dim}(A)$ is the number of dimensions of $A$.

\section{System Model}\label{sec:sysmodel}
Consider an antenna located at the coordinate $\mathbf{t}$ in a set of coordinate points, $A$. Varying the dimensions of $A$ allows a range of antenna scenarios to be considered, from a fixed antenna in 0D to a fluid antenna (FA) able to move to any point within a cuboid in 3D. We consider four antenna structures, detailed in Table \ref{tab:dims} and visually illustrated in Fig. \ref{fig:FAS_layout}.

\begin{table}[ht]
\caption{CFAS layouts considered for dimensions 0 to 3.}
\begin{tabular}{|c|c|c|}
\hline
$\mathrm{dim}(\!A)$ & Scenario & Definitions \\ \hline
0D & Fixed antenna & $\mathbf{t}=0$ \\ \hline
1D & FA on a line of length $T_1$ & $\mathbf{t}=t\in A= [0,T_1]$ \\ \hline
2D & \begin{tabular}[c]{@{}c@{}} FA in a rectangle with\\ side lengths $T_1$ and $T_2$.\end{tabular}& \begin{tabular}[c]{@{}c@{}}$\mathbf{t}=[t_1,t_2]^T\in A$ \\ $A= [0,T_1]\times[0,T_2]$ \end{tabular}\\ \hline
3D & \begin{tabular}[c]{@{}c@{}} FA in a cuboid with side\\  lengths $T_1$, $T_2$ and $T_3$.\end{tabular} & \begin{tabular}[c]{@{}c@{}}$\mathbf{t}=[t_1,t_2,t_3]^T \in A$ \\ $\!A\!=\![0,\!T_1]\!\times\![0,\!T_2]\!\times\![0,\!T_3]$\end{tabular} \\ \hline
\end{tabular}
\label{tab:dims}
\end{table}

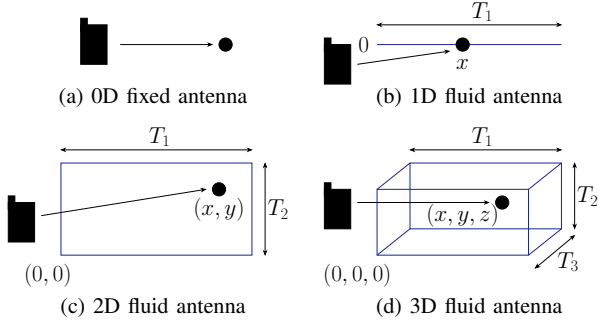
\begin{figure}[!ht]
\centering
\resizebox{0.9\textwidth}{!}{%
\begin{circuitikz}
\tikzstyle{every node}=[font=\Huge]
\draw [ color={rgb,255:red,17; green,39; blue,146} , line width=1pt ] (3,16.75) rectangle (10.25,13.25);
\draw [ fill={rgb,255:red,0; green,0; blue,0} ] (9,15.75) circle (0.25cm);
\node [font=\Huge] at (9,15) {$(x,y)$};
\node [font=\Huge] at (2.5,12.5) {$(0,0)$};
\draw [<->, >=Stealth] (3,17.25) -- (10.25,17.25);
\draw [<->, >=Stealth] (10.75,16.75) -- (10.75,13.25);
\node [font=\Huge] at (6.75,17.75) {$T_1$};
\node [font=\Huge] at (11.25,15) {$T_2$};
\draw [ fill={rgb,255:red,0; green,0; blue,0} ] (1,15.25) rectangle (2,13.75);
\draw [ fill={rgb,255:red,0; green,0; blue,0} ] (1,15) rectangle (1.25,15.5);
\draw [->, >=Stealth] (2.25,14.75) -- (8.5,15.75);
\draw [ color={rgb,255:red,17; green,39; blue,146} , line width=1pt ] (15,15.75) rectangle (20.75,13.25);
\draw [ color={rgb,255:red,17; green,39; blue,146} , line width=1pt ] (16.25,16.75) rectangle (22,14.25);
\draw [ color={rgb,255:red,32; green,35; blue,142}, line width=1pt, short] (15,15.75) -- (16.25,16.75);
\draw [ color={rgb,255:red,32; green,35; blue,142}, line width=1pt, short] (20.75,15.75) -- (22,16.75);
\draw [ color={rgb,255:red,32; green,35; blue,142}, line width=1pt, short] (15,13.25) -- (16.25,14.25);
\draw [ color={rgb,255:red,32; green,35; blue,142}, line width=1pt, short] (20.75,13.25) -- (22,14.25);
\draw [<->, >=Stealth] (16.25,17.25) -- (22,17.25);
\draw [<->, >=Stealth] (22.5,16.75) -- (22.5,14.25);
\draw [<->, >=Stealth] (22.5,14) -- (21,12.75);
\node [font=\Huge] at (23,15.5) {$T_2$};
\node [font=\Huge] at (19,17.75) {$T_1$};
\node [font=\Huge] at (22.25,13) {$T_3$};
\draw [ fill={rgb,255:red,0; green,0; blue,0} ] (19.75,15.25) circle (0.25cm);
\node [font=\Huge] at (18.25,14.75) {$(x,y,z)$};
\draw [ fill={rgb,255:red,0; green,0; blue,0} ] (13,15.75) rectangle (14,14.25);
\draw [ fill={rgb,255:red,0; green,0; blue,0} ] (13,15.5) rectangle (13.25,16);
\draw [->, >=Stealth] (14.25,15.25) -- (19.25,15.25);
\node [font=\Huge] at (14.25,12.5) {$(0,0,0)$};
\draw [ color={rgb,255:red,32; green,35; blue,142} , line width=1pt ] (15,21.25) rectangle (22,21.25);
\draw [ fill={rgb,255:red,0; green,0; blue,0} ] (18.25,21.25) circle (0.25cm);
\draw [<->, >=Stealth] (15,22) -- (22,22);
\node [font=\Huge] at (19,22.5) {$T_1$};
\node [font=\Huge] at (18.25,20.5) {$x$};
\draw [ fill={rgb,255:red,0; green,0; blue,0} ] (13,21.25) rectangle (14,19.75);
\draw [ fill={rgb,255:red,0; green,0; blue,0} ] (13,21) rectangle (13.25,21.5);
\draw [->, >=Stealth] (14.25,20.5) -- (17.75,21);
\node [font=\Huge] at (14.5,21.25) {$0$};
\node [font=\Huge] at (18.5,19.25) {(b) 1D fluid antenna};
\node [font=\Huge] at (6.5,11.25) {(c) 2D fluid antenna};
\node [font=\Huge] at (18.5,11.25) {(d) 3D fluid antenna};
\draw [ fill={rgb,255:red,0; green,0; blue,0} ] (3.75,22) rectangle (4.75,20.5);
\draw [ fill={rgb,255:red,0; green,0; blue,0} ] (3.75,21.75) rectangle (4,22.25);
\draw [->, >=Stealth] (5.25,21.25) -- (8.75,21.25);
\node [font=\LARGE] at (18.75,20.75) {};
\node [font=\LARGE] at (18.75,20.75) {};
\draw [ fill={rgb,255:red,0; green,0; blue,0} ] (9.25,21.25) circle (0.25cm);
\node [font=\Huge] at (6.5,19.25) {(a) 0D fixed antenna
};
\end{circuitikz}
}%
\caption{Fluid antenna geometries for 0 to 3 dimensions.}
\label{fig:FAS_layout}
\end{figure}
\vspace{0.75em}
We consider a correlated Ricean channel from a single antenna source to an $n\in\{0,1,2,3\}$-dimensional FA, denoted by $h^{(n)}(\mathbf{t})$. This channel is comprised of a rank-1 LoS
component and a correlated Rayleigh component, such that
\begin{align}
    h^{(n)}(\mathbf{t}) &= h^{(n)}_\mathrm{LoS}(\mathbf{t}) + h^{(n)}_\mathrm{NLoS}(\mathbf{t}) \notag\\
    & = \sqrt{\frac{\beta\,\kappa}{\kappa+1}}\tilde{h}^{(n)}_\mathrm{LoS}(\mathbf{t}) + \sqrt{\frac{\beta}{2(1+\kappa)}}\tilde{h}^{(n)}_\mathrm{NLoS}(\mathbf{t}) \label{eq:h},
\end{align}
where $\beta$ is the channel gain, $\kappa$ is the Ricean K-factor, $\tilde{h}^{(n)}_\mathrm{LoS}(\mathbf{t})=a^{(n)}(\mathbf{t})$, $\tilde{h}^{(n)}_\mathrm{NLoS}(\mathbf{t})\sim\mathcal{CN}(0,2)$, and $a^{(n)}(\mathbf{t})$ is a LoS steering function of the form $\exp{(j\mathbf{c}_n^T\mathbf{t})}$, where
\begin{align}
    {a}^{(0)}(\mathbf{t}) &= 1, \notag \\
    {a}^{(1)}(\mathbf{t}) &= \mathrm{e}^{2\pi jt\sin(\phi)\sin(\theta)},\notag \\
    {a}^{(2)}(\mathbf{t}) &= \mathrm{e}^{2\pi j(t_1\sin(\phi)\sin(\theta) + t_2\cos(\theta))},\notag \\
    {a}^{(3)}(\mathbf{t}) &= \mathrm{e}^{2\pi j(t_1\cos(\phi)\sin(\theta) + t_2\sin(\phi)\sin(\theta) +t_3\cos(\theta))}.\notag
\end{align}
The angles in the steering vectors are the angles of azimuth, $\phi$,  and elevation, $\theta$, of the LoS ray. Note that the real and imaginary components of $h^{(n)}_\mathrm{NLoS}(\mathbf{t})$ both have unit power to simplify the use of random field theory in Section \ref{sec:analysis}. We assume that the set $A$ is small enough so that $\beta$ is constant across $A$. Therefore, the received signal at $\mathbf{t}$ is \begin{equation}\label{rx}
    r(\mathbf{t}) = h^{(n)}(\mathbf{t})s + \nu,
\end{equation}
where $s$ is the transmitted signal with $\mathbb{E}[|s|^2]=E_s$ and $\nu\sim\mathcal{CN}(0,\sigma^2)$ is the additive white Gaussian noise.

As motivated in Section \ref{sec:intro}, in this work, we consider a CFAS employing MF that can move to the point with the highest SNR in a multidimensional continuous space. From \eqref{eq:h} and \eqref{rx}, the optimal SNR for a perfect CFAS is
\begin{align}
    \mathrm{SNR}^{(n)}\!&=\!\frac{\beta E_s}{2(\kappa+1)\sigma^2}\sup_{\mathbf{t}\in A}\left\{|\sqrt{2\kappa}\tilde{h}^{(n)}_\mathrm{LoS}(\mathbf{t}) +\tilde{h}^{(n)}_\mathrm{NLoS}(\mathbf{t})|^2\!\right\}\notag\\
    &=\frac{\beta E_s}{2(\kappa+1)\sigma^2}\sup_{\mathbf{t}\in A}\left\{X^{(n)}(\mathbf{t})\right\}\!,
\end{align}
where $X^{(n)}(\mathbf{t})$ is a noncentral chi-squared process ($\chi_2^{2}(\lambda)$) with noncentrality parameter $\lambda = 2\kappa$.


The probability of an SNR exceeding some high threshold (HSP), is denoted as $P^{(n)}_{hs} = P(\mathrm{SNR}^{(n)}>u)$. Defining the normalized threshold, $x=\frac{2(\kappa+1)\sigma^2 u}{\beta E_s}$, gives
\begin{equation}
    P^{(n)}_{hs} = P\left(\sup_{\mathbf{t}\in A}\left\{X^{(n)}(\mathbf{t})\right\} \geq x\right). \label{eq:Ppp}
\end{equation}
Hence, the HSP is converted to a problem relating to the supremum of a noncentral $\chi_2^2$ process.

The correlation of $\tilde{h}_\mathrm{NLoS}^{(n)}(\mathbf{t})$ over $A$ is a key physical factor affecting the performance of a CFAS system. As is conventional in the FAS space \cite{kkwong1,kkwong2,FAS_6G,psomas}, we assume isotropic correlation of the form
\begin{equation}
    \rho(\tau) = \frac{1}{2}\mathbb{E}\left[\tilde{h}^{(n)}_\mathrm{NLoS}(\mathbf{t})\tilde{h}^{(n)*}_\mathrm{NLoS}(\mathbf{t+\Delta})\right],
\end{equation}
where $\tau=||\mathbf{\Delta}||$ is the Euclidean separation between $\mathbf{t}$ and $\mathbf{t+\Delta}$. Additionally, we assume that 
\begin{equation}\label{lambda2}
\rho(\tau)\sim 1 -a\tau^2 \quad \mathrm{as} \quad \tau \rightarrow 0, 
\end{equation}
to ensure a mean-square differentiable channel. Distances in this work (i.e, $\tau$, $T_i$, $i=1,2,3$) are measured in wavelengths.
\vspace{-0.75em}
\section{Analysis}\label{sec:analysis}
In this section, we outline the known HSP result for a fixed antenna in Ricean fading and use existing methods to derive the HSP for a 1D CFAS experiencing Ricean fading. However, the methods used do not allow for an extension to higher dimensions, so we develop an accurate and much more general approximation to the HSP for CFASs using random field theory that can be applied to all scenarios in Table \ref{tab:dims}. These results are valid for arbitrary channel parameter values.
\vspace{-0.75em}
\subsection{Zero Dimensions (Fixed Antenna)}
For a fixed antenna, $X^{(0)}(t)=X^{(0)}(0)$ as $t\in A = \{0\}$ is a fixed point. Therefore, the HSP is the complementary CDF (CCDF) of a noncentral $\chi^2_2$ process, which is known to be
\begin{equation}
    P_{hs}^{(0)} =  Q_1\big(\sqrt{\lambda}, \sqrt{x}\big)= Q_1\big(\sqrt{2\kappa}, \sqrt{x}\big). 
    \label{eq:0D}
\end{equation}
This is a simple Ricean fading result at a single point.

\subsection{One Dimension}
\label{sec:1D}
For the case of a 1D CFAS, where the antenna can move to any point along a line of length $T_1$, the HSP is
\begin{equation}\label{eq1d}
    P_{hs}^{(1)} = P\bigg(\sup_{{t}\in [0,T_1]}\left\{X^{(1)}({t})\right\} \geq x\bigg).
\end{equation}
One way to find an asymptotically exact solution for \eqref{eq1d} is to use the method detailed in \cite{davies}. This is applied to the HSP for a Rayleigh fading environment in Appendix C of \cite{psomas}.  For the Ricean case, this method gives

\begin{equation}
\label{eq:1DLCR}
    \hat{P}_{hs}^{(1)} =P\big(X^{(1)}(0)>x\big) \!+ T_1\times\mathrm{LCR}(x),
\end{equation}
where $\mathrm{LCR}(x)$ is the level crossing rate (LCR) of $X^{(1)}(t)$ across $x$ and $\hat{P}_{hs}^{(n)}$ denotes an asymptotic approximation of ${P}_{hs}^{(n)}$ based on LCR theory. From \cite{cheng_second_2009}, the LCR is given by
\begin{align}
    \label{eq:LCR}
    \!\!\mathrm{LCR}(x)&\!= 2\sqrt{\frac{x}{\pi}}\mathrm{e}^{-(\kappa+x/2)} \!\int_0^{\pi/2}\!\cosh\!\left(\sqrt{2\kappa x}\cos\!\left(\psi\right)\!\right) \notag \\
    &\times\! \Big(\!\mathrm{e}^{2\kappa\sin^2(\phi)\sin^2(\theta)\sin^2(\psi)}\!+\!\sqrt{2\kappa\pi}j\sin(\phi)  \sin(\theta)\notag \\ &\times\!\sin(\psi)\,\!\mathrm{erf}\!\left(\!\sqrt{2\kappa}j \sin(\phi)\sin(\theta)\sin(\psi)\!\right)\!\!\Big)d\psi.
 \end{align}
Substituting \eqref{eq:LCR} and the well-known Ricean result,  $P(X^{(1)}(0)>x)=Q_1(\sqrt{2\kappa}, \sqrt{x})$, into \eqref{eq:1DLCR} gives the HSP approximation as desired.

While this method is intuitive for 1D, the integral cannot be solved analytically, and the concept of level crossings is not easily extended to higher dimensions. Therefore, we must use a different method to derive closed-form expressions for the HSP in higher dimensions.

\subsection{Arbitrary Numbers of Dimensions}
\label{subsec:arbitrarydims}
Commonly used in random field theory and applied in a Rayleigh fading CFAS context in \cite{smith_dimensional_2025}, the EEC is a useful tool for approximating ${P_{hs}^{(n)}}$. In the 1D scenario in Section \ref{sec:1D}, the LCR counts the number of threshold exceedances across a line of length $T_1$. The EEC is a more general topological implementation of this concept, providing a normalized count of exceedance regions across an arbitrary-dimensional space. 

In the Rayleigh fading work in \cite{smith_dimensional_2025}, the EEC for a $\chi^2_2$ process was shown to give an asymptotically exact approximation for $P_{hs}^{(n)}$ as $u\rightarrow\infty$. For a Ricean fading scenario, we instead apply the EEC for a noncentral $\chi^2_2$ process. However, an important assumption of the EEC in its standard form is that the underlying random field has stationary statistics and thus its mean and covariance structure do not vary in space. Such processes are of the form 
\begin{equation}\label{Y}
    Y(\mathbf{t}) = |\mu + u(\mathbf{t})|^2,
\end{equation}
where $\mu$ is constant and $u(\mathbf{t})$ is a spatially correlated zero-mean Gaussian process with independent mean and imaginary parts. Now, the Ricean channel in \eqref{eq:h} leads to the $X^{(n)}(\mathbf{t})$ process, with the slightly different form
\begin{equation}\label{Xoriginal}
    X^{(n)}(\mathbf{t}) = |\mu(\mathbf{t}) + u(\mathbf{t})|^2, 
\end{equation}
where the mean, $\mu(\mathbf{t}) = \sqrt{2\kappa} a^{(n)}(\mathbf{t})=\sqrt{2\kappa}\exp{(j\mathbf{c}_n^T\mathbf{t})}$, has a constant amplitude but the phase varies with the FA position. All other assumptions required for the random field theory are satisfied. Rearranging \eqref{Xoriginal} gives the alternative version
\begin{equation}\label{Xnew}
    X^{(n)}(\mathbf{t}) = |\sqrt{2\kappa}+\exp{(-j\mathbf{c}_n^T\mathbf{t})} u(\mathbf{t})|^2.
\end{equation}
In \eqref{Xnew}, the LoS term is constant (i.e. in the correct form), but the modified Gaussian process, $\exp{(-j\mathbf{c}_n^T\mathbf{t})} u(\mathbf{t})$, has a correlation structure that does not exactly match the assumptions in \eqref{Y} due to the mixing of real and imaginary parts. Hence, the precise assumptions of random field theory are closely but not exactly followed.

To motivate the continued use of the EEC approximation in this case, we note that \eqref{Xoriginal} satisfies all the assumptions except for a non-constant LoS phase, and that all assumptions are satisfied at broadside (as $h_\mathrm{LoS}$ is constant) and when $\kappa=0$. Additionally, numerical results in Section \ref{sec:numresults} show the accuracy of the approach.

From \cite[eq. 15.10.1]{adler}, the EEC is defined as
\begin{equation}
    \mathrm{EEC} = \sum_{j=0}^{\mathrm{dim}(A)} L_j(A)\rho_j(x), \label{eq:EEC}
\end{equation}
where $L_j(A)$ are the Lipschitz-Killing curvatures of $A$ and $\rho_j(x)$ are the Euler Characteristic (EC) densities. The EEC is used to approximate the HSP in the normal way and this approximation is denoted by $\tilde{P}_{hs}^{(n)} =\textrm{EEC}$, where $n=\mathrm{dim}(A)$. The following result from Lemma 5.1 of \cite{taylor_gaussian_2006} gives the EC densities of a noncentral $\chi^2_2$.

\begin{res}
\label{res:ECdensities}
The EC densities for a noncentral $\chi_2^{2}(\lambda)$ random field are given by
\begin{equation}
    \rho_0(x) = P\left(\chi_2^{2}(\lambda) \geq x\right), \label{eq:rho0}
\end{equation}
\vspace{-1em}
\begin{multline}
\label{eq:rhoj}
    \rho_j(\lambda,x)\!=\!\frac{\mathrm{e}^{-\frac{\lambda+x}{2}}}{(2\pi)^{j/2}}\!\sum\limits_{i=0}^\infty\!\frac{\lambda^ix^{1+i-j/2}}{4^ii!\Gamma(1\!+\!i)}\!\sum\limits_{l=0}^{\lfloor\!\frac{j\!-\!1}{2}\!\rfloor}\sum\limits_{m=0}^{j\!-\!1\!-\!2l}\!\!\!1_{\{\!2\geq j-m-2l-2i\!\}}\\ \times\binom{1+2i}{j\!-\!1\!-\!m\!-\!2l\!}\!\frac{(-1)^{j-1+m+l}(j\!-\!1)!x^{m+l}}{m!\,l!\,2^l}\!.
\end{multline}
\end{res}

Although exact, the infinite sum in \eqref{eq:rhoj} complicates computation and interpretation. Therefore, we reformulate \eqref{eq:rhoj} into a simpler closed-form expression in the following Lemma. To our knowledge, this provides the first closed-form expression for the EC density of a noncentral $\chi_2^2$ random field, enabling a more direct analysis of the HSP.

\begin{lem}
\label{lem:genform}
    The result in \eqref{eq:rhoj} can be simplified to the closed-form expression 
    \begin{align}
    \label{eq:rhojj}
    \rho_j(&\lambda,x)\!=\!\frac{\mathrm{e}^{-\frac{\lambda+x}{2}}x^{\frac{1-j}{2}}(j\!-\!1)!}{(2\pi)^{j/2}}\sum\limits_{l=0}^{\lfloor\frac{j-1}{2}\rfloor}\sum\limits_{m=0}^{j-1-2l}\frac{(-1)^{j-1+m+l}}{m!\,l!\,r!\,2^{l+r-1}} \notag \\ &\times x^{m+l+r/2}\lambda^{r/2} \bigg(\!\frac{\sqrt{\lambda x}}{2}I_r\left(\sqrt{\lambda x}\right)+\sum\limits_{t=0}^{r-1}\bigg(r\binom{r-1}{t} \notag \\ & \times I_{2t-r+1}\left(\sqrt{\lambda x}\right)+\frac{\sqrt{\lambda x}}{2}\binom{r}{t}I_{2t-r}\left(\sqrt{\lambda x}\right)\!\bigg)\!\bigg).
    \end{align}
\end{lem}
\begin{proof}
    See Appendix A.
\end{proof}
To compute the Killing-Lipschitz curvatures, we use the result in \cite[p. 324, p. 333]{adler},
\begin{equation}
    L_j(A) = \lambda_2^{j/2}L_j^E(A), \label{eq:LjA}
\end{equation}
where $L_j^E(A)$ are the Euclidean intrinsic volumes defined in Table \ref{tab:euclideanintrinsicvalues} \cite{adlerpdf} and $\lambda_2$ is the variance of the channel derivative. In isotropic fading, the derivative can be taken in any spatial direction. When $\rho(\tau)=J_0(2\pi\tau)$ as considered in this work, $\lambda_2=2\pi^2$ \cite{foschini}. We now derive the key results of the paper.

\vspace{-0.4em}
\begin{table}[ht!]
\centering
\caption{Euclidean Intrinsic Volumes}
\label{tab:euclideanintrinsicvalues}
\renewcommand{\arraystretch}{1.1}
\begin{tabular}{|c|c|c|}
\hline
\multicolumn{1}{|c|}{Dimension} & \multicolumn{1}{c|}{Intrinsic Volume}  & \multicolumn{1}{c|}{Physical Meaning} \\ \hline
0D & $L_0^E(A)=1$ & \\ \hline
1D & \begin{tabular}[c]{@{}c@{}}$L_0^E(A)=1$\\ $L_1^E(A)=T_1$\end{tabular} & \begin{tabular}[c]{@{}c@{}} \\ Length \end{tabular} \\ \hline
2D & \begin{tabular}[c]{@{}c@{}}$L_0^E(A)=1$\\ $L_1^E(A)=T_1+T_2$ \\ $L_2^E(A)=T_1T_2$ \end{tabular} &  \begin{tabular}[c]{@{}c@{}} \\ $\tfrac{1}{2}\!\times$Boundary Length \\ Area \end{tabular} \\ \hline
3D & \begin{tabular}[c]{@{}c@{}}$L_0^E(A)=1$\\ $L_1^E(A)=T_1+T_2+T_3$ \\ $L_2^E\!(A)\!\!=\!\!T_1T_2\!+\!T_1T_3\!+\!T_2T_3$ \\ $L_3^E(A)=T_1T_2T_3$ \end{tabular} & \begin{tabular}[c]{@{}c@{}} \\ $2\times$Caliper Diameter \\ $\tfrac{1}{2}\times\,$Surface Area \\ Volume \end{tabular} \\ \hline 
\end{tabular}
\end{table}

\begin{lem}
\label{lem:HSP}
In 0 dimensions (a fixed antenna),
\begin{equation}
    P_{hs}^{(0)} = Q_1\left(\sqrt{2\kappa}, \sqrt{x}\right). \label{eq:0D2}
\end{equation}
In 1 dimension,
\begin{equation}
    \tilde{P}_{hs}^{(1)}\!\approx\! Q_1\left(\!\sqrt{2\kappa}, \sqrt{x}\right)\!+ \mathrm{e}^{-(\kappa+x/2)}\,T_1\sqrt{\frac{\lambda_2x}{2\pi}}I_0\!\left(\!\sqrt{2\kappa x}\right)\!.
    \label{eq:1D}
\end{equation}
In 2 dimensions,
\begin{align}
\label{eq:2D}
    & \tilde{P}_{hs}^{(2)}\!\approx Q_1\!\!\left(\!\sqrt{2\kappa},\! \sqrt{x}\right)\!+\! \mathrm{e}^{-(\kappa+x/2)}\!\sqrt{\!\frac{\lambda_2}{2\pi}}\! \bigg[\!\bigg(\!\!(T_1\!+\!T_2)\sqrt{x}\!+\! T_1T_2\notag \\ &\,\,\times\!\sqrt{\!\frac{\lambda_2}{2\pi}}(x\!-\!1)\!\!\bigg) I_0\!\left(\!\sqrt{2\kappa x}\right)\!-\!T_1T_2\sqrt{\frac{\!\lambda_2\kappa x}{\pi}}I_1\!\left(\!\sqrt{2\kappa x}\right)\!\! \bigg]\!.
\end{align}
The 3D HSP, $ \tilde{P}_{hs}^{(3)}$, is stated in \eqref{eq:3D} at the bottom of the page.
\begin{figure*}[b]
\vspace{-0.5em}
\hrulefill
\normalsize
\begin{multline}
\label{eq:3D}
     \tilde{P}_{hs}^{(3)}\!\approx Q_1\!\!\left(\!\sqrt{2\kappa}, \!\sqrt{x}\right)\!+\!\mathrm{e}^{-(\kappa+x/2)}\sqrt{\frac{\lambda_2}{2\pi}}\bigg[\!\bigg(\!(T_1\!+\!T_2\!+\!T_3)\sqrt{x}\!+\!\sqrt{\frac{\lambda_2}{2\pi}}\bigg(\!(T_1T_2\!+\!T_1T_3\!+\!T_2T_3)(x\!-\!1)\!+\!T_1T_2T_3\sqrt{\frac{\lambda_2x}{2\pi}}(x\!+\!\kappa\!-\!3)\!\bigg)\!\!\bigg) \\ \times\! I_0\!\left(\!\sqrt{2\kappa x}\right)\!+ \!\sqrt{\frac{\kappa\lambda_2}{\pi}}\bigg(\!T_1T_2T_3\sqrt{\frac{2\lambda_2}{\pi}}(1-x)\!-\!(T_1T_2\!+\!T_1T_3\!+\!T_2T_3)\sqrt{x}\bigg)I_1\!\left(\!\sqrt{2\kappa x}\right)\!+\!T_1T_2T_3\frac{\lambda_2\kappa\sqrt{x}}{2\pi}I_2\!\left(\!\sqrt{2\kappa x}\right)\!\!\bigg].
\end{multline}
\end{figure*}
\end{lem}
\vspace{-1.75em}
\begin{proof}
Combining \eqref{eq:EEC} and the EC density stated in \eqref{eq:rho0} and defined in \eqref{eq:0D} gives the result for $P_{hs}^{(0)}$ in \eqref{eq:0D2}. Combining \eqref{eq:EEC} and \eqref{eq:rhojj} for the relevant $\mathrm{dim}(A)$, substituting $\lambda = 2\kappa$ and simplifying the resulting expression gives $ \tilde{P}_{hs}^{(1)},  \tilde{P}_{hs}^{(2)}$ and $ \tilde{P}_{hs}^{(3)}$ in \eqref{eq:1D}-\eqref{eq:3D}, respectively.
\end{proof}

Remarkably, despite the infinite summation in the EC density expression in \eqref{eq:rhoj}, the HSP for all four dimensions considered can be simplified to a closed-form expression of elementary functions, the well known modified Bessel function and a single Marcum-Q function. The equations defined are not only powerful tools to approximate an otherwise intractable problem, but are also fast and simple to compute.  

As conceptually discussed earlier in this section, the approximate result for the 1D scenario in \eqref{eq:1D} is asymptotically exact when the antenna is broadside to the UE ($\phi=0$). Here, the LCR expression in \eqref{eq:LCR} simplifies to
\begin{align}
   \underset{\phi=0}{\mathrm{LCR}}(x)&=2\sqrt{\frac{x}{\pi}}\mathrm{e}^{-(\kappa+x/2)}\!\int_0^{\pi/2}\!\cosh\!\left(\sqrt{2\kappa x}\cos\!\left(\psi\right)\right)d\psi,\! \notag \\
   &=\sqrt{x}\,\mathrm{e}^{-(\kappa+x/2)}I_0\!\left(\sqrt{2\kappa x}\right). \label{eq:LCRphi0}
\end{align}
Combining \eqref{eq:1DLCR} and \eqref{eq:LCRphi0} yields the same result as \eqref{eq:1D}. This arises because, at $\phi=0$, the incoming signal is perpendicular to all points in the FAS, thereby eliminating phase differences between locations. Consequently, $X^{(n)}(\mathbf{t})$  takes the form in \eqref{Y}, satisfying the assumptions of random field theory underlying the EEC method.  This leads to the same result as the asymptotically exact LCR method. Similarly, setting $\kappa=0$ in \eqref{eq:1DLCR} and \eqref{eq:1D} gives an identical result to the Rayleigh fading result in \cite[eq. 19]{smith_dimensional_2025} as expected.

\section{Numerical Results}\label{sec:numresults}
This section validates the analytical results and examines the system behavior. As outlined in Section II, the correlation over distance \( \tau \) follows the classical Jakes’ model, where \( \rho(\tau) = J_0(2\pi\tau) \) \cite{foschini}. While this model is used in simulations, the analysis applies to all correlation models of the form \( \rho(\tau) \sim 1 - a\tau^2 \) as \( \tau \to 0 \). In simulations, $\frac{\beta E_s}{\sigma^2}{=}1$ for simplicity, and channels are generated at intervals of $0.01\lambda$ in each dimension within the available antenna space to determine the location with the highest SNR, with \( 10^6 \) replicates produced.
\subsection{High SNR Probability for an $n$-dimensional FA}
Figure \ref{fig:alldims} validates the analytical HSP results in Lemma \ref{lem:HSP}. The logarithmically scaled complementary CDF is plotted for a CFAS with $n=\{0, 1, 2, 3\}$ dimensions, where $T=T_1=T_2=T_3=0.25\lambda$.  As in \cite{smith_dimensional_2025}, the small value of $T=0.25\lambda$ was used for these results due to the computational challenges of generating channels that sufficiently sample the antenna space in the simulated 3D scenario. This highlights the relevance and usefulness of the analytical expressions. For all dimensions, $\phi=\frac{\pi}{4}$ and both $\kappa=0.2$ and 2 are considered.

\begin{figure}[ht]
    \centering
    \vspace{-0.25em}
    \includegraphics[scale=0.55]{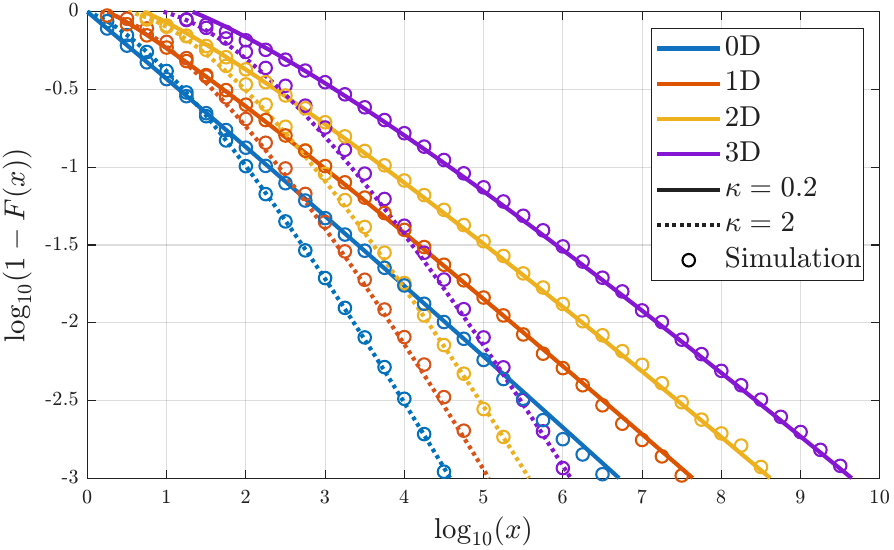}
    \caption{A comparison of the analytical and simulated HSPs for 0-3 dimension CFASs, where each dimension is set to 0.25$\lambda$.}
    \label{fig:alldims}
\end{figure}

Figure \ref{fig:alldims} demonstrates excellent agreement between the simulated results and the analytical HSP expressions in \eqref{eq:0D2}-\eqref{eq:3D}, validating the effectiveness of the EEC approach. The results show that $\log_{10}(P_{hs})$ increases approximately linearly with the number of dimensions. Therefore, the $n$-th dimension scales the HSP by a scalar multiplicative factor. As expected, the HSP performance deteriorates significantly as $\kappa$ increases. The 3D HSP is 0.031 at $\kappa=0.2$ but 0.0012 at $\kappa=2$, an order of magnitude lower. A stronger LoS component makes the channel more deterministic, reducing the spatial variations and consequently limiting the performance gains achievable at the optimal position. As a result, the performance improvements from increasing the number of dimensions are more pronounced for lower $\kappa$, where each additional dimension introduces greater channel variability.
\vspace{-0.2em}
\subsection{Relationship between FA size and Ricean K-factor}
Figure \ref{fig:size} examines the increase in the area of a 2D square FAS required for a Ricean environment to achieve an HSP performance equivalent to that of a 2D square FAS in a Rayleigh environment, illustrating the compensation in size needed to counteract the reduced fading diversity. We numerically solve the expression $\tilde{P}_{hs,\mathrm{Rice}}^{(2)}-\tilde{P}_{hs,\mathrm{Ray}}^{(2)} =0$ for $T_\mathrm{Rice}$, where $\tilde{P}_{hs,\mathrm{Rice}}^{(2)}$ is defined in \eqref{eq:2D} with $T_\mathrm{Rice}=T_1=T_2$ and 
\begin{multline}
    \tilde{P}_{hs,\mathrm{Ray}}^{(2)}= Q_1\!\left(0,\sqrt{x_0}\right) + \mathrm{e}^{-x_0/2} \\ \times \bigg[\!\sqrt{\frac{2\lambda_2}{\pi}}T_\mathrm{Ray}\sqrt{x_0}+T_\mathrm{Ray}^2\frac{\lambda_2}{2\pi}(2x_0-1)\bigg],
\end{multline}
which is \eqref{eq:2D} with $\kappa=0$ and $T_\mathrm{Ray}=T_1=T_2$. We assume $T_\mathrm{Ray}$ is fixed, $x_0 = \frac{2\sigma^2 u}{\beta E_s}$ and $x$ is set so the HSP is 0.01.
\vspace{0.4em}
\begin{figure}[t]
    \centering
    \includegraphics[scale=0.55]{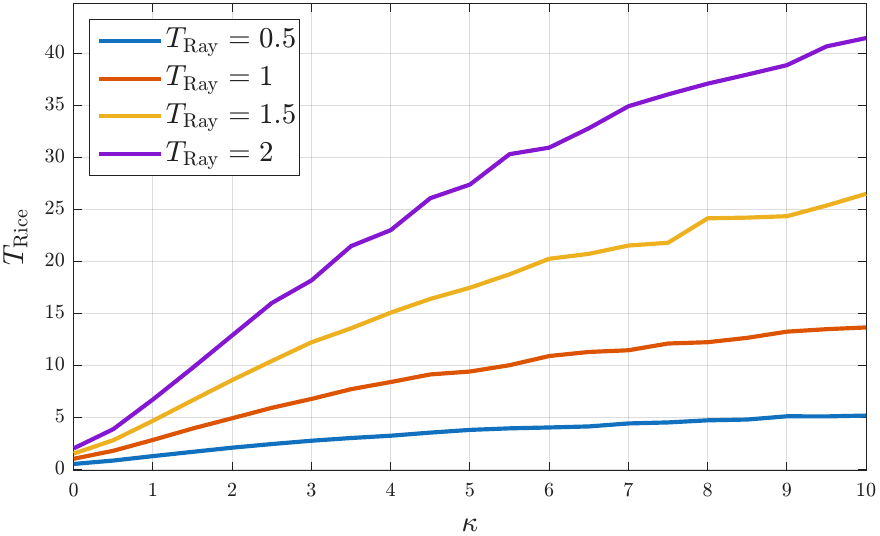}
    \caption{The side length \( T_\mathrm{Rice} \) required for a CFAS in a Ricean environment to match the HSP of 0.01 of a CFAS in a Rayleigh environment with \( T_\mathrm{Ray} \).}
    \label{fig:size}
\end{figure}

As in Fig. \ref{fig:alldims}, Fig. \ref{fig:size} also highlights how influential $\kappa$ is to the FAS performance in a Ricean fading environment. It is clear that as $\kappa$ increases, a significant increase in the side length of the FAS is required to maintain an HSP performance equivalent to that experienced under Rayleigh fading. As discussed above, this is to compensate for reduced channel variations. Table \ref{tab:areainc} shows the ratio between the area of a FAS in a Ricean environment ($A_\mathrm{Rice}=T_\mathrm{Rice}^2$) and a Rayleigh environment ($A_\mathrm{Ray}=T_\mathrm{Ray}^2$) required for an HSP of 0.01.
\vspace{-0.3em}
\begin{table}[ht]
\caption{Ratio between the Ricean and Rayleigh FAS areas ($A_\mathrm{Rice}/A_\mathrm{Ray}$) required for a HSP of 0.01}
\begin{tabular}{|c|c|c|c|}
\hline
$T_\mathrm{Ray}$ & $\kappa=1$ & $\kappa=4$ & $\kappa=7$ \\ \hline
0.5 & 6.25 & 41.42 & 77.44 \\ \hline
1 & 7.90 & 70.39 & 131.10 \\ \hline
1.5 & 9.65 & 100.88 & 205.83 \\ \hline
2 & 11.25 & 132.37 & 305.03\\ \hline
\end{tabular}
\label{tab:areainc}
\end{table}
\vspace{0.5em}

Figure \ref{fig:alldims} and Table \ref{tab:areainc} show that $A_\mathrm{Rice}/A_\mathrm{Ray}$ increases with $T_0$, the FAS side length in a Rayleigh environment. This occurs because the HSP performance is already high when $T_\mathrm{Ray}$ is large, requiring a faster increase in FAS area to introduce sufficient variation to achieve comparable performance.

\subsection{Comparison of the 1D EEC and LCR Methods}

While Fig. \ref{fig:alldims} shows the accuracy of the EEC method when the azimuthal AoA is $\phi=\frac{\pi}{4}$ and $\kappa=\{0.2,2\}$, we are interested in investigating the accuracy in more detail.  As discussed in Section \ref{subsec:arbitrarydims}, the LCR method is the only other analytical method for the HSP of a FAS and is asymptotically exact, but is only applicable to 1D FASs. Therefore, Fig. \ref{fig:diff} considers the difference between the EEC and LCR methods in the 1D scenario for a range of K-factors and azimuthal AoAs. As in Fig. \ref{fig:size}, $x$ is set to give an HSP of 0.01.

\begin{figure}[t]
    \centering
    \includegraphics[scale=0.53,trim={0.05cm 0.1cm 0cm 0cm},clip]{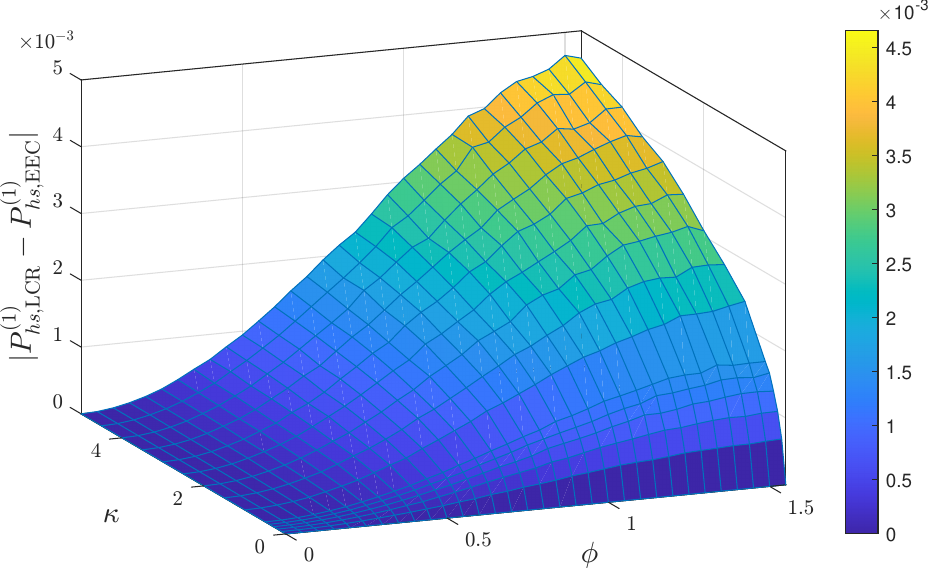}
    \caption{The difference between the 1D HSP using the LCR and EEC methods.}
    \label{fig:diff}
\end{figure}

When \( \kappa = 0 \) and \( \phi = 0 \), the two methods yield identical results, confirming the discussion in Section \ref{subsec:arbitrarydims}. This occurs because, in these scenarios, the LoS channel maintains a constant phase across the entire FAS, satisfying all assumptions necessary for the EEC method. Notably, the EEC method exhibits the highest accuracy in the most practically relevant cases. The accuracy deteriorates as endfire (\(\phi = \frac{\pi}{2}\)) is approached, but this is an undesirable system configuration. Also, accuracy decreases as $\kappa$ increases. As shown in Figs. \ref{fig:alldims} and \ref{fig:size}, the FAS becomes less effective in these conditions due to reduced channel variations, so such scenarios are unlikely candidates for FAS deployment. In the preferred operating regions (low \( \kappa \), \( \phi < \frac{\pi}{4} \)), the EEC method closely approximates the asymptotically exact LCR method.
\vspace{-0.2em}

\section{Conclusion}

In this work, we presented the first analytical results for a CFAS in a Ricean fading environment. We derived a highly accurate approximation for the HSP of a CFAS in 0, 1, 2, and 3 dimensions using the EEC. We also provided the first closed-form expression for the EC density of a non-central \(\chi_2^2\) random field. The approximate HSPs exhibit excellent agreement with simulations and closely match the asymptotically exact LCR method HSP in 1D. Finally, we analyzed the impact of the Ricean K-factor on the HSP of a CFAS, demonstrating that performance degrades rapidly as $\kappa$ increases, underscoring the effectiveness of FAS in rich scattering environments.

\vspace{-0.2em}
\section*{Appendix A \\ Proof of Lemma \ref{lem:genform}}

Using $\Gamma(1+i)=i!$ and rearranging \eqref{eq:rhoj} gives:
\begin{align}
    \!\!\rho_j(x)&\!=\!\frac{\mathrm{e}^{-\frac{\lambda+x}{2}}x^{1-j/2}(j-1)!}{(2\pi)^{j/2}}\sum\limits_{l=0}^{\lfloor\frac{j-1}{2}\rfloor}\sum\limits_{m=0}^{j-1-2l}\frac{(-1)^{j-1+m+l}}{l!\,m!\,2^l} \notag\\ &\quad\times x^{m+l}\sum\limits_{i=0}^\infty \frac{\left(\frac{\lambda x}{4}\right)^{\!i}}{(i!)^2}1_{\{2\geq j-m-2l-2i\}}\!\binom{1+2i}{j\!-\!1\!-\!m\!-\!2l}\!, \notag \\
    &\!=\!\frac{\mathrm{e}^{-\frac{\lambda+x}{2}}\!x^{1\!-\!j/2}(j\!\!-\!\!1\!)!}{(2\pi)^{j/2}}\!\!\sum\limits_{l=0}^{\lfloor\!\frac{j-1}{2}\!\rfloor}\!\sum\limits_{m=0}^{j\!-\!1\!-\!2l}\!\!\frac{(\!-1\!)^{j\!-\!1\!+m+l}x^{m+l}\!\!}{l!\,m!\,2^l}S_{jlm}(x), \label{eq:rhojT} 
\end{align} \vspace{-0.5em}
where
\begin{equation}
    S_{jlm}(x)\!=\!\sum\limits_{i=0}^\infty \frac{\left(\frac{\lambda x}{4}\right)^i}{(i!)^2}1_{\{2\geq j-m-2l-2i\}}\binom{1+2i}{j\!-\!1\!-\!m\!-\!2l}\!.
\end{equation}
Note that the outer summations in \eqref{eq:rhojT} are finite. Focusing on $S_{jlm}(x)$, the innermost infinite summation over $i$, let $z=\frac{\lambda x}{4}$ and $r=j\!-\!1\!-\!m\!-\!2l$. Thus, 
\begin{equation}\label{indfn}
    S_{jlm}(x) = \sum\limits_{i=0}^\infty \frac{z^i}{(i!)^2}1_{\{2\geq r+1-2i\}}\binom{1+2i}{r}.
\end{equation}
Formally expanding the  binomial coefficient in \eqref{indfn} gives
\begin{equation}\label{binco}
\binom{1+2i}{r} =\frac{(2i-r+2)(2i-r+3) \ldots (2i+1)}{r!}.
\end{equation}
Note that when the indicator function in \eqref{indfn} returns 1, all terms in the numerator of \eqref{binco} are positive and when the indicator function returns zero, one of the terms in the numerator is also zero. As a result, the indicator function is redundant and can be removed if the binomial coefficient is replaced by \eqref{binco}. Hence, letting $y=\sqrt{z}$,
\begin{equation}\label{Sjlm}
    S_{jlm}(x)\!=\!\frac{y^{r-1}}{r!}\sum\limits_{i=0}^\infty\frac{y^{2i-r+1}}{(i!)^2}\left((2i\!+\!1)\dots(2i\!-\!r\!+\!2)\right).
\end{equation}
As $\frac{d^r}{dy^r}y^{2i+1}\!\!=\!(2i\!+\!1)\dots(2i\!-\!r\!+\!2)y^{2i-r+1}\!$, \eqref{Sjlm} collapses to
\begin{equation}
    S_{jlm}(x)=\frac{y^{r-1}}{r!}\frac{d^r}{dy^r}\left\{y\,I_0\!\left(2y\right)\right\}. \label{eq:Tdr}
\end{equation}
Applying the Leibniz rule,
\begin{equation}
    \frac{d^r}{dy^r}\{y\,I_0\!\left(2y\right)\}\!=\! \sum\limits_{s=0}^{r}\binom{r}{s}\!\!\left(\frac{d^s}{dy^s}y\right)\!\!\left(\frac{d^{r-s}}{dy^{r-s}}I_0\!\left(2y\right)\right).
\end{equation}
The derivatives of $y$ are $0$ when $s>1$. Therefore, only the $s=\{0,1\}$-th derivatives of $y$ need to be considered. From \cite{nist_digital_2023}, the $r$-th derivative of $I_0(2y)$ is $\sum_{s=0}^r\binom{r}{s}I_{2s-r}(2y)$. Therefore,
\begin{multline}
    \frac{d^r}{dy^r}\left\{y\,I_0\!\left(2y\right)\right\}=y\,I_r(2y) + \sum\limits_{s=0}^{r-1}\bigg(r\binom{r-1}{s} \\ \times I_{2s-r+1}(2y)+y\binom{r}{s}I_{2s-r}(2y)\bigg). \label{eq:dry}
\end{multline}
Combining \eqref{eq:rhojT}, \eqref{eq:Tdr} and \eqref{eq:dry} gives \eqref{eq:rhojj}.  
\bibliographystyle{IEEEtran}
\bibliography{IEEEabrv, referencesFAS}

\begin{thebibliography}{10}
\providecommand{\url}[1]{#1}
\csname url@samestyle\endcsname
\providecommand{\newblock}{\relax}
\providecommand{\bibinfo}[2]{#2}
\providecommand{\BIBentrySTDinterwordspacing}{\spaceskip=0pt\relax}
\providecommand{\BIBentryALTinterwordstretchfactor}{4}
\providecommand{\BIBentryALTinterwordspacing}{\spaceskip=\fontdimen2\font plus
\BIBentryALTinterwordstretchfactor\fontdimen3\font minus \fontdimen4\font\relax}
\providecommand{\BIBforeignlanguage}[2]{{%
\expandafter\ifx\csname l@#1\endcsname\relax
\typeout{** WARNING: IEEEtran.bst: No hyphenation pattern has been}%
\typeout{** loaded for the language `#1'. Using the pattern for}%
\typeout{** the default language instead.}%
\else
\language=\csname l@#1\endcsname
\fi
#2}}
\providecommand{\BIBdecl}{\relax}
\BIBdecl

\bibitem{kkwong1}
K.-K. Wong, A.~Shojaeifard, K.-F. Tong, and Y.~Zhang, ``Fluid antenna systems,'' \emph{IEEE Trans. Commun.}, vol.~20, no.~3, pp. 1950--1962, Mar. 2021.

\bibitem{kkwong2}
------, ``Performance limits of fluid antenna systems,'' \emph{IEEE Commun. Lett.}, vol.~24, no.~11, pp. 2469--2472, Nov. 2020.

\bibitem{FAS_6G}
W.~K. New \emph{et~al.}, ``A tutorial on fluid antenna system for {6G} networks: Encompassing communication theory, optimization methods and hardware designs,'' \emph{IEEE Commun. Surv. Tutor.}, early access, Nov. 2024.

\bibitem{zhu_moveable_2024}
L.~Zhu, W.~Ma, and R.~Zhang, ``Movable antennas for wireless communication: {O}pportunities and challenges,'' \emph{IEEE Commun. Mag.}, vol.~62, no.~6, pp. 114--120, June 2024.

\bibitem{zhu_modeling_2024}
------, ``Modeling and performance analysis for movable antenna enabled wireless communications,'' \emph{IEEE Trans. Wirel. Commun.}, vol.~23, no.~6, pp. 6234--6250, June 2024.

\bibitem{psomas}
C.~Psomas, P.~J. Smith, H.~A. Suraweera, and I.~Krikidis, ``Continuous fluid antenna systems: Modeling and analysis,'' \emph{IEEE Commun. Lett.}, vol.~27, no.~12, pp. 3370--3374, Dec. 2023.

\bibitem{smith_dimensional_2025}
P.~J. Smith, A.~S. Inwood, M.~Matthaiou, and R.~Senanayake, ``Dimensional scaling laws for continuous fluid antenna systems,'' \emph{IEEE Wirel. Commun. Lett.}, (to appear).

\bibitem{RISFAS}
J.~Chen \emph{et~al.}, ``Low-complexity beamforming design for {RIS}-assisted fluid antenna systems,'' in \emph{Proc. IEEE Int. Conf. Commun.}, June 2024, pp. 1377--1382.

\bibitem{ULFAS}
G.~Hu \emph{et~al.}, ``Fluid antennas-enabled multiuser uplink: {A} low-complexity gradient descent for total transmit power minimization,'' \emph{IEEE Commun. Lett.}, vol.~28, no.~3, pp. 602--606, Mar. 2024.

\bibitem{alouini}
M.~Khammassi, A.~Kammoun, and M.-S. Alouini, ``A new analytical approximation of the fluid antenna system channel,'' \emph{IEEE Trans. Wirel. Commun.}, vol.~22, no.~12, pp. 8843--8858, Dec. 2023.

\bibitem{block}
P.~Ramírez-Espinosa, D.~Morales-Jimenez, and K.-K. Wong, ``A new spatial block-correlation model for fluid antenna systems,'' \emph{IEEE Trans. Wirel. Commun.}, vol.~23, no.~11, pp. 15\,829--15\,843, Nov. 2024.

\bibitem{cop1}
F.~Rostami~Ghadi, K.-K. Wong, F.~J. López-Martínez, and K.-F. Tong, ``Copula-based performance analysis for fluid antenna systems under arbitrary fading channels,'' \emph{IEEE Commun. Lett.}, vol.~27, no.~11, pp. 3068--3072, Nov. 2023.

\bibitem{cop2}
F.~Rostami~Ghadi \emph{et~al.}, ``A {G}aussian copula approach to the performance analysis of fluid antenna systems,'' \emph{IEEE Trans. Wirel. Commun.}, vol.~23, no.~11, pp. 17\,573--17\,585, Nov. 2024.

\bibitem{adler}
R.~J. Adler and J.~E. Taylor, \emph{Random {F}ields and {G}eometry}.\hskip 1em plus 0.5em minus 0.4em\relax New York, NY, USA: Springer Science \& Business Media, 2009.

\bibitem{davies}
R.~B. Davies, ``Hypothesis testing when a nuisance parameter is present only under the alternatives,'' \emph{Biometrika}, vol.~74, no.~1, pp. 33--43, Mar. 1987.

\bibitem{cheng_second_2009}
X.~Cheng, C.-X. Wang, D.~I. Laurenson, and A.~V. Vasilakos, ``Second order statistics of non-isotropic mobile-to-mobile {R}icean fading channels,'' in \emph{Proc. IEEE Int. Conf. Commun.}, Jun. 2009.

\bibitem{taylor_gaussian_2006}
J.~E. Taylor, ``A {G}aussian kinematic formula,'' \emph{The Annals of Probability}, vol.~34, no.~1, pp. 122--158, Jan. 2006.

\bibitem{adlerpdf}
R.~J. Adler, J.~E. Taylor, and K.~J. Worsley, ``Applications of {R}andom {F}ields and {G}eometry: Foundations and {C}ase {S}tudies,'' {A}ccessed: Mar. 27, 2025. [Online]. Available: https://robert.net.technion.ac.il/publications/.

\bibitem{foschini}
D.-S. Shiu, G.~Foschini, M.~Gans, and J.~Kahn, ``Fading correlation and its effect on the capacity of multielement antenna systems,'' \emph{IEEE Trans. Commun.}, vol.~48, no.~3, pp. 502--513, Mar. 2000.

\bibitem{nist_digital_2023}
\BIBentryALTinterwordspacing
{National Institute of Standards and Technology}, ``Digital {L}ibrary of {M}athematical {F}unctions.'' {A}ccessed: Mar. 18, 2025. [Online]. Available: \url{https://dlmf.nist.gov}
\BIBentrySTDinterwordspacing

\end{thebibliography}

\end{document}